\documentclass
[preprint,byrevtex,pre,nofootinbib,10pt,balancelastpage,titlepage,bibnotes,twocolumn]{revtex4}%
\usepackage{amsfonts}
\usepackage{amsmath}
\usepackage{amssymb}
\usepackage{graphicx}%
\providecommand{\U}[1]{\protect\rule{.1in}{.1in}}
\newtheorem{theorem}{Theorem}

\newenvironment{proof}[1][Proof]{\noindent\textbf{#1.} }{\ \rule{0.5em}{0.5em}}
\ifx\pdfoutput\relax\let\pdfoutput=\undefined\fi
\newcount\msipdfoutput
\ifx\pdfoutput\undefined\else
\ifcase\pdfoutput\else
\ifx\paperwidth\undefined\else
\ifdim\paperheight=0pt\relax\else\pdfpageheight\paperheight\fi
\ifdim\paperwidth=0pt\relax\else\pdfpagewidth\paperwidth\fi
\fi\fi\fi
\begin{document}
\preprint{UATP/1701}
\title{Nonequilibrium Work and its Hamiltonian Connection for a Microstate in
Nonequilibrium Statistical Thermodynamics: A \ Case of Mistaken Identity}
\author{P.D. Gujrati}
\email{pdg@uakron.edu}
\affiliation{Department of Physics, Department of Polymer Science, The University of Akron,
Akron, OH 44325}

\begin{abstract}
Nonequilibrium work-Hamiltonian connection for a microstate plays a central
role in diverse branches of statistical thermodynamics (fluctuation theorems,
quantum thermodynamics, stochastic thermodynamics, etc.). We show that the
change in the Hamiltonian for a microstate should be identified with the work
done by it, and not the work done \emph{on} it. This contradicts the current
practice in the field. The difference represents a contribution whose average
gives the work that is dissipated due to irreversibility. As the latter has
been overlooked, the current identification does not properly account for
irreversibilty. As an example, we show that the corrected version of
Jarzynski's relation can be applied to free expansion, where the original
relation fails. Thus, the correction has far-reaching consequences and
requires reassessment of current applications.

\end{abstract}
\date{January 31, 2017}
\maketitle

The work-Hamiltonian relations at the level of \emph{microstates} play a
central role in the application of the first law to a system $\Sigma$ under
nonequilibrium conditions in a medium $\widetilde{\Sigma}$ in diverse branches
of nonequilibrium (NEQ) statistical thermodynamics including but not limited
to \emph{nonequilibrium work theorems}
\cite{Bochkov,Jarzynski,Crooks,Pitaevskii}, \emph{stochastic}
\emph{thermodynamics} \cite{Sekimoto,Seifert} and \emph{quantum
thermodynamics} \cite{Lebowitz,Alicki}. Once work is identified, heat is
identified by invoking the first law. Unfortunately, this endeavor has given
rise to a controversy about the actual meaning of work, which apparently is
far from settled
\cite{Cohen,Jarzynski-Cohen,Sung,Gross,Jarzynski-Gross,Peliti,Rubi,Jarzynski-Rubi,Rubi-Jarzynski,Peliti-Rubi,Rubi-Peliti,Pitaevskii,Bochkov}%
. The controversy is distinct from the confusion about the meaning of work and
heat in classical nonequilibrium thermodynamics \cite[for example, and
references therein]{Fermi,Gislason,Kestin} involving a \emph{system-intrinsic
}(SI)\emph{ }or\emph{ medium-intrinsic} (MI)\ description, the latter only
recently been clarified
\cite{Gujrati-Heat-Work0,Gujrati-Heat-Work,Gujrati-I,Gujrati-II,Gujrati-III,Gujrati-Entropy1,Gujrati-Entropy2}%
. For example, the classical thermodynamic formulations of work in the two
descriptions take the form $PdV$ or\emph{ }$P_{0}dV$, respectively, in terms
of the instantaneous pressure $P$ of $\Sigma$ or $P_{0}~$of $\widetilde{\Sigma
}$ and their volume change $dV$ or $-dV$, respectively. As usual, $\Sigma$ and
$\widetilde{\Sigma}$ form an isolated system $\Sigma_{0}$. Throughout this
work, we will assume that $\Sigma$ and $\widetilde{\Sigma}$ are statistically
quasi-independent \cite{Gujrati-II,Gujrati-Entropy2} so that their entropies
are additive.

In equilibrium (EQ), the state of $\Sigma$ is described by its set of
extensive observables, or state variables, such as its energy $E$, volume $V$,
the number of particles $N$, etc. We collectively denote them by $\mathbf{X}$.
They are controllable from outside by an observer and appear as parameters in
the Hamiltonian $\mathcal{H}$ of $\Sigma$. The equilibrium entropy is a state
function $S(\mathbf{X})$ of $\mathbf{X}$. Away from equilibrium, we also need
an additional set of extensive internal variables $\mathbf{\xi}$
\cite{Gujrati-I,Gujrati-II,Gujrati-Entropy2,deGroot,Prigogine,Maugin} to
specify the state of a NEQ system so that $\mathbf{Z=(X},\mathbf{\xi)}$
denotes the set of state variables. The internal variables are not
controllable from outside by an observer. A NEQ state for which its entropy is
a state function $S(\mathbf{Z})$ of $\mathbf{Z}$ is said to be an
\emph{internal equilibrium} (IEQ) state \cite{Gujrati-I,Gujrati-II}; if not,
the entropy $S(\mathbf{Z},t)$ is an explicit function of time $t$. If we do
not use $\mathbf{\xi}$ to specify a NEQ state, then the entropy $S(\mathbf{X}%
,t)$ again is not a state function even if we have an IEQ state. In any NEQ
state, there are going to be internal processes that are beyond the control of
an observer, but have to be accounted for a proper thermodynamic description.

The MI-description always refers to exchange quantities (work or heat, to be
denoted here by their modern notation $d_{\text{e}}W$ or $\allowbreak
d_{\text{e}}Q$ \cite{note-0,deGroot,Prigogine}, respectively) between the
medium and the system and involve quantities referring to the medium and are
readily identifiable and measurable. No internal variables are required since
their affinities vanish for the medium as it is always in equilibrium. The
SI-description always refers to quantities (work or heat, to be denoted here
by $dW$ or $\allowbreak dQ$, respectively) intrinsic to the system,
\textit{i.e.} they contain quantities pertaining to the system alone. These
quantities represent SI quantities and may include internal variables as their
affinities do not vanish for $\Sigma$; in many cases, they may not be readily
measurable or even identifiable and require care in interpreting results.
Therefore, the use of exchange quantities is quite widespread. Despite this,
we have concluded
\cite{Gujrati-Heat-Work0,Gujrati-Heat-Work,Gujrati-I,Gujrati-II,Gujrati-III,Gujrati-Entropy1,Gujrati-Entropy2}
that the SI-description is more appropriate to study nonequilibrium processes,
even if we do not use $\mathbf{\xi}$, since $\mathcal{H}$ plays a central role
for dynamics.

Traditional formulation of nonequilibrium statistical mechanics and
thermodynamics \cite{Landau,Gibbs} starts with a mechanistic approach in which
$\Sigma$, or more precisely its microstate $\mathcal{M}$, follows its
classical or quantum mechanical evolution in time, which will require focusing
on the Hamiltonian $\mathcal{H}$ of the system, a SI quantity; the interaction
with $\widetilde{\Sigma}$ is usually treated as a very weak stochastic
perturbation on it. This immediately suggests adopting a SI-description.
Unfortunately, this description has been overlooked by the current
practitioners in the field who have consistently used a MI-description such as
in Jarzynski's nonequilibrium work relation%
\begin{equation}
\left\langle e^{-\beta_{0}\Delta\widetilde{W}_{k}^{\prime}}\right\rangle
_{0}\overset{?}{=}e^{-\beta_{0}\Delta F^{\prime}}~\ \ (=\left\langle
e^{\beta_{0}\Delta W_{k}^{\prime}}\right\rangle _{0})\label{JarzynskiRelation}%
\end{equation}
to be explained below. In Eqs. (\ref{JarzynskiRelation}%
-\ref{Av-ExternalWork-Medium}), the question mark on the first (and currently
widely accepted) equality implies that it may be questionable in all cases. We
will establish that \emph{an appreciation of the SI-description will paved the
way for a correct work-Hamiltonian relation for a microstate }resulting in the
last equality in Eqs. (\ref{JarzynskiRelation}-\ref{Av-ExternalWork-Medium})
in all cases. This is the main point of this paper. We then draw attention to
some of its consequences. We have applied this approach to the set $\left\{
\mathcal{M}_{k}\right\}  $ of microstates
\cite{Gujrati-Heat-Work0,Gujrati-Heat-Work,Gujrati-Entropy1,Gujrati-Entropy2}
to obtain a \emph{microscopic representation} of work and heat in terms of the
set of microstate probabilities $\left\{  p_{k}\right\}  $, which will be
exploited here. As we will be dealing with microstates, we will mostly use
their energy set $\left\{  E_{k}\right\}  $ instead of $\mathcal{H}$ in the following.

In Eq. (\ref{JarzynskiRelation}), $\left\langle {}\right\rangle _{0}$ refers
to averaging with respect to the canonical probability distribution $\left\{
e^{-\beta_{0}E_{k}^{\prime}}/Z_{\text{in}}^{\prime}(\beta_{0})\right\}  $ of
the initial equilibrium state A$_{\text{eq}}$, $Z_{\text{in}}^{\prime}%
(\beta_{0})$ the initial equilibrium partition function for the system,
$\Delta\widetilde{W}_{k}^{\prime}$ the work done \emph{on} the $k$th
microstate by the external working medium $\widetilde{\Sigma}$ during a
process $\gamma$ (not necessarily reversible) connecting the terminal states,
both at the same inverse temperature $\beta_{0}$, and $\Delta F^{\prime}$ is
the change in the free energy between those states. (The mystery behind the
prime will become clear later.) However, the average $\left\langle
{}\right\rangle _{0}$ in Eq. (\ref{JarzynskiRelation}) does not represent a
thermodynamic average over $\gamma$. The inverse temperature along $\gamma$
may not always exist or may be different than $\beta_{0}$ due to
irreversibility \cite{Cohen,Muschik}.

The differential external work $d\widetilde{W}_{k}$ done \emph{on} the $k$th
microstate $\mathcal{M}_{k}$\ during a segment of the path $\gamma$ between
$t$\ and $t+dt$\ is identified as the change $dE_{k}\doteq E_{k}%
(t+dt)-E_{k}(t)$ in its energy $E_{k}(\mathbf{Z})$ as parameters in
$\mathbf{Z}$ change:%
\begin{equation}
d\widetilde{W}_{k}\overset{?}{=}\frac{\partial E_{k}}{\partial\mathbf{Z}%
}.d\mathbf{Z}\equiv dE_{k}\ \ \ (=-dW_{k})\text{.}\label{ExternalWork-Medium}%
\end{equation}
Observe that the microstate $\mathcal{M}_{k}$\ does not change during the
performance of the above work, only its energy changes as $\mathbf{Z}$
changes. The relation is common to all of the three domains of activities
\cite[ for example]{Jarzynski,Alicki,Sekimoto} except that $\mathbf{Z}$ is
replaced by $\mathbf{X}$. The net external work done on the system is the
integral over the path $\gamma$
\begin{equation}
\Delta\widetilde{W}_{k}\overset{?}{=}%
{\textstyle\int\nolimits_{\gamma}}
\frac{\partial E_{k}}{\partial\mathbf{Z}}.d\mathbf{Z=}%
{\textstyle\int\nolimits_{\gamma}}
dE_{k}\ \ \ (=-\Delta W_{k})\mathbf{.}\label{ExternalWork-Medium-Total}%
\end{equation}
The thermodynamic works $d\widetilde{W}$ and $\Delta\widetilde{W}$\ are
averages over $\left\{  p_{k}\right\}  $ (not to be confused with
$\left\langle {}\right\rangle _{0}$ in Eq. (\ref{JarzynskiRelation})):%
\begin{equation}
d\widetilde{W}\overset{?}{=}%
{\textstyle\sum\limits_{k}}
p_{k}dE_{k}~(=-dW),\Delta\widetilde{W}\overset{?}{=}%
{\textstyle\int\nolimits_{\gamma}}
{\textstyle\sum\limits_{k}}
p_{k}dE_{k}~(=-\Delta W)\mathbf{.}\label{Av-ExternalWork-Medium}%
\end{equation}

Before proceeding further, we introduce various notions of work that are
relevant here by two simple examples.

(1) Consider as our system a general but purely a classical mechanical
one-dimensional spring of arbitrary Hamiltonian $\mathcal{H}(x,p)$ with one
end fixed at an immobile wall and the other end of mass $m$ free to move. The
free end is pulled by an \emph{external} force (not necessarily a constant)
$F_{0}$ applied at time $t=0$. A microstate $\mathcal{M}$ can be thought of as
a small phase space area with its center at $k\doteq(x,p)$. However, we will
not show the index $k$ below for simplicity. Initially the spring is
undisturbed and has zero\emph{ }SI restoring spring force $F=-\partial
\mathcal{H}/\partial x$. The total force $F_{\text{t}}=F_{0}+F$ acts like the
force \emph{imbalance} $F_{\text{t}}\lessgtr0$. There is no mechanical
equilibrium unless $F_{\text{t}}=0$ and the spring continues to stretch or
contract. The\emph{ }SI work done by $F$ is identified as the\emph{ }work
$dW\doteq Fdx$ performed by the spring, while the work performed by $F_{0}$ is
identified as the work $d\widetilde{W}=F_{0}dx$ \emph{transferred} to the
spring; its negative $d_{\text{e}}W=-F_{0}dx$ is identified as the work
performed by the spring \emph{against} the external force. As this is a purely
mechanical example, there is no dissipation. Despite this, we can introduce
using the modern notation \cite{note-0}
\begin{equation}
d_{\text{i}}W\doteq dW-d_{\text{e}}W\equiv dW+d\widetilde{W}\equiv
F_{\text{t}}dx,\label{diW-microstate}%
\end{equation}
which can be of either sign and represents the work done by the imbalance
$F_{\text{t}}$. Thus, $dW,d_{\text{e}}W=-d\widetilde{W}$ and $d_{\text{i}}W$
represent \emph{different} works, a result that has nothing to do with
irreversibility but only with the imbalance; among these, only $dW$ is a SI
work. The change in the Hamiltonian $\mathcal{H}=E$ of the spring due to a
variation in the work variable $x$ is $\left.  d\mathcal{H}\right\vert
_{\text{w}}=\left.  dE\right\vert _{\text{w}}=Fdx=-dW\neq d\widetilde{W}$,
where we have used a suffix w to refer to the change caused by the performance
of work.

Above, we have considered the \emph{exclusive} Hamiltonian $\mathcal{H}$
\cite{Jarzynski}. Let us consider the \emph{inclusive} Hamiltonian
$\mathcal{H}^{\prime}=E^{\prime}\doteq E-F_{0}x$ \cite{Jarzynski} used in
deriving Eq. (\ref{JarzynskiRelation}); this explains the presence of a prime
there. For the inclusive energy $E^{\prime}$, $dE^{\prime}=dE-d(F_{0}%
x)=-F_{\text{t}}dx-xdF_{0}.$ As $\partial E^{\prime}/\partial x=-F_{\text{t}}$
does not identically vanish, $E^{\prime}(x,F_{0})$ is a function of \emph{two}
work parameters $x$ and $F_{0}$. As $\partial E^{\prime}/\partial F_{0}=-x$,
$x$ is the generalized force conjugate to $F_{0}$. The corresponding SI work
$dW^{\prime}$ consists of two contributions due to variations in $x$ and
$F_{0}$: $dW^{\prime}=dW_{x}^{\prime}+dW_{F_{0}}^{\prime}=F_{\text{t}%
}dx+xdF_{0}$ and satisfies $dW^{\prime}=-\left.  dE^{\prime}\right\vert
_{\text{w}}$ just as above with the exclusive\emph{ }Hamiltonian $\mathcal{H}%
$. Furthermore, the contribution $dW_{F_{0}}^{\prime}\equiv xdF_{0}$
represents the exchange work $d_{\text{e}}W^{\prime}$ with the medium so that
$d\widetilde{W}^{\prime}=-xdF_{0}\neq\left.  dE^{\prime}\right\vert
_{\text{w}}$ \cite{Jarzynski} represents the exchange work that appears in the
left side of Eq. (\ref{JarzynskiRelation}). The following identities are
always satisfied:%
\begin{equation}%
\begin{tabular}
[c]{c}%
$\left.  dE^{\prime}\right\vert _{\text{w}}-\left.  dE\right\vert _{\text{w}%
}\doteq dW-dW^{\prime}\equiv-d(F_{0}x),$\\
$d\widetilde{W}^{\prime}-d\widetilde{W}\doteq d_{\text{e}}W-d_{\text{e}%
}W^{\prime}\equiv-d(F_{0}x).$\\
$d_{\text{i}}W^{\prime}\equiv d_{\text{i}}W\equiv F_{\text{t}}dx.$%
\end{tabular}
\ \ \label{Inclusive-Exclusive-Change}%
\end{equation}

Let us investigate the case $F_{0}\equiv0$. In this case, $F_{\text{t}}=F$,
and $dW=dW^{\prime}=Fdx\neq0$  and $d\widetilde{W}=d\widetilde{W}^{\prime}=0$
as a consequence of $\mathcal{H}=\mathcal{H}^{\prime}$. Such a situation
arises when the spring, which is initially kept locked in a compressed (or
elongated) state is unlocked to let go without applying any external force.
Here, $dW=dW^{\prime}\neq0$, while $d\widetilde{W}=d\widetilde{W}^{\prime}=0$
as the spring expands (or contracts) under the influence of its spring force
$F$.

(2) To incorporate dissipation, we consider a thermodynamic analog of the
above example: a gas in a cylinder, closed at one end and a movable piston at
the other end. The piston is locked and the gas has a pressure $P$. We first
focus on various work averages to understand the form of dissipation. At time
$t=0$, an external pressure $P_{0}<P$ is applied on the piston and the lock on
the piston is released. We should formally make the substitution $x\rightarrow
V,F\rightarrow P$ (or $P_{k}$ when considering $\mathcal{M}_{k}$) and
$F_{0}\rightarrow-P_{0}$. The gas expands $(dV\geq0)$ and $P\searrow P_{0}$.
The SI work done by the gas is $dW=PdV$, while $d\widetilde{W}=-P_{0}%
dV=-d_{\text{e}}W$. The difference $d_{\text{i}}W\doteq dW-d_{\text{e}%
}W=(P-P_{0})dV\geq0$ appears as the work that is dissipated in the form of
heat ($d_{\text{i}}Q\equiv d_{\text{i}}W$\ as will be shown below) either due
to the friction between the piston and the cylinder or other dissipative
forces like the viscosity of the gas.

Let us analyze this model more carefully at a microstate level but without
using any $\mathbf{\xi}$ for the sake of simplicity. Let the Hamiltonian of
the gas be denoted by $\mathcal{H}(\mathbf{X})=\allowbreak E(\mathbf{X})$. In
the following, we will only show $V$ and keep all other parameters held fixed.
Therefore, the only work we will consider is due to the generalized force
conjugate to $V$. Let $E_{k}(V)$ denote the energy of some $\mathcal{M}_{k}$
and $P_{k}\doteq-\partial E_{k}/\partial V$; we have $E(V)\doteq\left\langle
E_{k}\right\rangle $ and $P(V)\doteq\left\langle P_{k}\right\rangle $. We
write $dE_{k}\doteq d_{\text{e}}E_{k}\mathcal{+}d_{\text{i}}E_{k}$
\cite{note-0}, and identify $dW_{k}\equiv-dE_{k}=P_{k}dV,d_{\text{e}}%
W_{k}\equiv-d_{\text{e}}E_{k}=P_{0}dV$ and $d_{\text{i}}W_{k}\equiv
-d_{\text{i}}E_{k}=(P_{k}-P_{0})dV$ as above, giving the three
work-Hamiltonian relations for a microstate. After statistical averaging, we
obtain $dW\equiv-dE=PdV,d_{\text{e}}W\equiv-\left.  d_{\text{e}}E\right\vert
_{\text{w}}=P_{0}dV$ and $d_{\text{i}}W\equiv-\left.  d_{\text{i}}E\right\vert
_{\text{w}}=(P-P_{0})dV$ just as discussed above; the suffix w means that
these averages are given by the first terms in Eq. (\ref{Av-dE}) below. 

As with the exclusive Hamiltonian above, the corresponding inclusive
Hamiltonian here for $\mathcal{M}_{k}$ takes the form of a NEQ enthalpy
\cite{Gujrati-II,Gujrati-III} $E_{k}^{\prime}(P_{0},V)=E_{k}(V)+P_{0}V$ and is
a function of $\mathbf{X}^{\prime}=(P_{0}$,$V)$. Therefore, $dW_{k}^{\prime
}=-\left(  \partial E_{k}^{\prime}/\partial V\right)  dV-\left(  \partial
E_{k}^{\prime}/\partial P_{0}\right)  dP_{0}=P_{k}dV-d(P_{0}V)$. We also have
$d\widetilde{W}_{k}=-P_{0}dV$ and $d\widetilde{W}_{k}^{\prime}=-VdP_{0}$ so
that Eq. (\ref{Inclusive-Exclusive-Change}) remains satisfied for each microstate.

The lesson from the two examples is that $dW_{k}$ or $dW$ is a SI quantity but
$d\widetilde{W}_{k}$ or $d\widetilde{W}$ is not. Similarly, microstate
energies $\left\{  E_{k}\right\}  $ and probabilities $\left\{  p_{k}\right\}
$, $\mathcal{H}$ or $d\mathcal{H}$, and the average energy $E$ or $dE$ of the
system are also SI quantities. These SI quantities play a very important role
in our discussion below. We now break $\left.  d\mathcal{H}\right\vert
_{\text{w}}=\left.  dE_{k}\right\vert _{\text{w}}$ for $\mathcal{M}_{k}$\ into
two parts \cite{note-0}:
\begin{equation}
\left.  d\mathcal{H}\right\vert _{\text{w}}=\left.  d_{\text{e}}%
\mathcal{H}\right\vert _{\text{w}}+\left.  d_{\text{i}}\mathcal{H}\right\vert
_{\text{w}}=-dW_{k},\label{Hamiltonian-Partition}%
\end{equation}
where $d_{\text{e}}\mathcal{H}=d_{\text{e}}E_{k}\doteq-d_{\text{e}}W_{k}\equiv
d\widetilde{W}_{k}$ is the change from the external work $d_{\text{e}}W$ and
$d_{\text{i}}\mathcal{H}=d_{\text{i}}E_{k}\doteq-d_{\text{i}}W_{k}$ is the
change due to the internal work $d_{\text{i}}W_{k}$ performed by the imbalance
in the generalized forces on $\mathcal{M}_{k}$. Thus, we make the following
claim in the form of a Theorem:

\begin{theorem}
\textbf{ }$\mathbf{Thermodynamic}$ $\mathbf{Work}$\textbf{-}$\mathbf{Energy}$
$\mathbf{Principle}$ \emph{The change }$\left.  d\mathcal{H}\right\vert
_{\text{w}}=dE_{k}$\emph{ in the Hamiltonian $\mathcal{H}$ due to work only
must be identified with the SI-work }$dW_{k}$\emph{ and not with
}$d\widetilde{W}_{k}$ for $\mathcal{M}_{k}$. \emph{It} \emph{has two
contributions as shown in Eq. (\ref{Hamiltonian-Partition}). The first one
corresponds to the external work }$d_{\text{e}}W_{k}=-d\widetilde{W}_{k}%
$\emph{ \textit{performed by }}$\mathcal{M}_{k}$\ \emph{\textit{against the
medium }and the second one to the internal work }$d_{\text{i}}W_{k}$
\emph{performed by the imbalance in the generalized forces. After\ a
statistical averaging over all microstates for a system, }$\left.
d_{\text{i}}E\right\vert _{\text{w}}\doteq\left\langle \left.  d_{\text{i}%
}\mathcal{H}\right\vert _{\text{w}}\right\rangle \equiv-d_{\text{i}}W\doteq-%
{\textstyle\sum\nolimits_{k}}
p_{k}d_{\text{i}}W_{k}\leq0$\emph{ results in dissipation in the system with
the inequality referring to irreversibility. }
\end{theorem}

\begin{proof}
Based on the two examples, the proof is almost trivial. For the average energy
$E$, we have
\begin{equation}
dE\equiv%
{\textstyle\sum\nolimits_{k}}
p_{k}dE_{k}+%
{\textstyle\sum\nolimits_{k}}
E_{k}dp_{k}.\label{Energy Differential}%
\end{equation}
As $\left\{  p_{k}\right\}  $ is not changed in the first sum, it is evaluated
at \emph{fixed entropy }of $\Sigma$. This isentropic sum $-dW$ is an average
of $-dW_{k}=dE_{k}\doteq(\partial E_{k}/\partial\mathbf{Z})\cdot
d\mathbf{Z}\equiv\allowbreak d_{\text{e}}E_{k}+d_{\text{i}}E_{k}$. This proves
the first part. As $\left\{  E_{k}\right\}  $ is unchanged in the second sum,
it refers to an isometric process at fixed $\mathbf{Z}$ and represents
generalized heat $dQ=d_{\text{e}}Q+d_{\text{i}}Q$. Thus, $dE=dQ-dW$. Since
$dE=d_{\text{e}}Q-d_{\text{e}}W$ also, we $d_{\text{i}}W\equiv d_{\text{i}}Q$.
To prove the last part, we turn to thermodynamics. It can be shown that the
temperature of any thermodynamic state can be defined by $dQ=TdS$
\cite{Gujrati-II,Gujrati-Heat-Work0,Gujrati-Entropy2} so we can write
$dE=TdS-dW$. We rewrite this as $dE=T_{0}d_{\text{e}}S-d_{\text{e}}%
W+T_{0}d_{\text{i}}S+(T-T_{0})dS-d_{\text{i}}W$, which leads to $T_{0}%
d_{\text{i}}S=(T_{0}-T)dS+d_{\text{i}}W$. Each term on the right side, being
independent of each other, must be nonnegative separately to ensure the second
law ($d_{\text{i}}S\geq0$). This proves the last part.
\end{proof}

It should be stressed that $-(\partial E_{k}/\partial\mathbf{Z})$ or
$-(\partial E_{k}^{\prime}/\partial\mathbf{Z}^{\prime})$ represents the
"generalized force" and the work has the conventional form:\ "force"$\times
$"distance," contrary to what is commonly stated. According to the claim, we
must use $-\Delta W^{\prime},-dW_{k},-\Delta W_{k}$ and $-dW,-\Delta W$ on the
left sides in Eqs. (\ref{JarzynskiRelation}-\ref{Av-ExternalWork-Medium}),
respectively, as shown by the enclosed parentheses.

Let us rewrite $dE$ as follows: $dW=TdS-dE=-dF+(T-T_{0})dS$, where
$F=E-T_{0}S$. As shown in the proof above, $(T-T_{0})dS\leq0$. Thus, we
conclude that $dW\leq-dF$. It is also easy to see that $d\widetilde{W}%
=dF+T_{0}d_{\text{i}}S\geq dF$. 

Using the partition of $dE_{k}$ in $dW_{k}$, we have $d_{\text{e}}W=-%
{\textstyle\sum\nolimits_{k}}
p_{k}d_{\text{e}}E_{k},$ \ $d_{\text{i}}W=-%
{\textstyle\sum\nolimits_{k}}
p_{k}d_{\text{i}}E_{k}\geq0$. Similarly, using the partition $dp_{k}\equiv
d_{\text{e}}p_{k}+d_{\text{i}}p_{k}$, we have $d_{\text{e}}Q=%
{\textstyle\sum\nolimits_{k}}
E_{k}d_{\text{e}}p_{k},$ \ $d_{\text{i}}Q=%
{\textstyle\sum\nolimits_{k}}
E_{k}d_{\text{i}}p_{k}\geq0$, and $d_{\text{e}}S=-%
{\textstyle\sum\nolimits_{k}}
\ln p_{k}d_{\text{e}}p_{k},$\ \ $d_{\text{i}}S=-%
{\textstyle\sum\nolimits_{k}}
\ln p_{k}d_{\text{i}}p_{k}\geq0$ as sum over microstates. We finally have%
\begin{align}
\Delta_{\text{e}}W &  =-%
{\textstyle\sum\nolimits_{k}}
{\textstyle\int\nolimits_{\gamma}}
p_{k}d_{\text{e}}E_{k},\Delta_{\text{i}}W=-%
{\textstyle\sum\nolimits_{k}}
{\textstyle\int\nolimits_{\gamma}}
p_{k}d_{\text{i}}E_{k}\geq0,\nonumber\\
\Delta_{\text{e}}Q &  =%
{\textstyle\sum\nolimits_{k}}
{\textstyle\int\nolimits_{\gamma}}
E_{k}d_{\text{e}}p_{k},\ \Delta_{\text{i}}Q=%
{\textstyle\sum\nolimits_{k}}
{\textstyle\int\nolimits_{\gamma}}
E_{k}d_{\text{i}}p_{k}\geq0,\label{CorrectedForms}\\
\Delta W &  =-%
{\textstyle\sum\nolimits_{k}}
{\textstyle\int\nolimits_{\gamma}}
p_{k}dE_{k},\Delta Q=%
{\textstyle\sum\nolimits_{k}}
{\textstyle\int\nolimits_{\gamma}}
E_{k}dp_{k},\nonumber
\end{align}
along with $\Delta_{\text{i}}W=\Delta_{\text{i}}Q$. The above equations
provide the correct identification of all quantities on the right sides of the
first equations in Eqs. (\ref{JarzynskiRelation}-\ref{Av-ExternalWork-Medium})
in terms of the right closed parentheses, and should be used in trajectory
thermodynamics or quantum thermodynamics. Again, using Eq.
(\ref{Hamiltonian-Partition}), we have (using$\ d_{\text{i}}W=d_{\text{i}}Q$
in the top equation)
\begin{subequations}
\label{Av-dE}%
\begin{align}
d_{\text{i}}E &  \doteq%
{\textstyle\sum\nolimits_{k}}
p_{k}d_{\text{i}}E_{k}+%
{\textstyle\sum\nolimits_{k}}
E_{k}d_{\text{i}}p_{k}=0,\label{Av-diE}\\
dE &  =d_{\text{e}}E\doteq%
{\textstyle\sum\nolimits_{k}}
p_{k}d_{\text{e}}E_{k}+%
{\textstyle\sum\nolimits_{k}}
E_{k}d_{\text{e}}p_{k}.\label{Av-deE}%
\end{align}
Even if $d_{\text{i}}E_{k}\neq0$, $d_{\text{i}}E=0$; thus, $E$ cannot change
by internal processes as is well know. The second equation gives the
conventional form of the second law in terms of the exchange quantities:
$dE=d_{\text{e}}E\equiv d_{\text{e}}Q-d_{\text{e}}W$.

Let us evaluate the particular average considered in Eq.
(\ref{JarzynskiRelation}) but of $e^{\beta_{0}\Delta W}$%
\end{subequations}
\[
\left\langle e^{\beta_{0}\Delta W}\right\rangle _{0}\doteq%
{\textstyle\sum\limits_{k}}
\frac{e^{-\beta_{0}E_{k,\text{in}}}}{Z_{\text{in}}(\beta_{0})}e^{\beta
_{0}\Delta W_{k}}=%
{\textstyle\sum\limits_{k}}
\frac{e^{-\beta_{0}E_{k,\text{in}}}}{Z_{\text{in}}(\beta_{0})}e^{-\beta
_{0}\left.  \Delta E_{k}\right\vert _{\text{w}}}.
\]
As the terminal states are equilibrium states, we have $\left.  \Delta
E_{k}\right\vert _{\text{w}}=E_{k,\text{fn}}-E_{k,\text{in}}$. Therefore,%
\begin{equation}
\left\langle e^{\beta_{0}\Delta W}\right\rangle _{0}=%
{\textstyle\sum\limits_{k}}
\frac{e^{-\beta_{0}E_{k,\text{fn}}}}{Z_{\text{in}}(\beta_{0})}=\frac
{Z_{\text{fn}}(\beta_{0})}{Z_{\text{in}}(\beta_{0})}=e^{-\beta_{0}\Delta F};
\label{GujratiRelation}%
\end{equation}
here $F=E-T_{0}S$ for the exclusive Hamiltonian. The same calculation can be
carried out for the inclusive Hamiltonian, with a similar result except all
the quantities must be replaced by their prime analog. This fixes the first
equation, the original Jarzynski relation, in Eq. (\ref{JarzynskiRelation}) by
the second equation in enclosed parentheses.

Let us apply Eq. (\ref{GujratiRelation}) to the example of a free expansion of
a one-dimensional ideal gas of classical particles, but treated quantum
mechanically as a particle in a box with rigid walls. We assume that the gas
is thermalized initially at some temperature $T_{0}$. It is isolated from the
medium so that the free expansion occurs in an isolated system. After the
expansion from the box size $L_{\text{in}}$ to $L_{\text{fn}}>L_{\text{in}}$,
the box is again thermalized at the same temperature $T_{0}$. As discussed at
the end of the first example for $F_{0}=0$, we note that $dW_{k}\neq0$ even
though $d\widetilde{W}_{k}=0$. Since we are dealing with an ideal gas, we can
focus on a single particle whose energy levels are in appropriate units
$E_{k}=k^{2}/L^{2}$, where $L$ is the length of the box. The change is due
only to the work as no heat is allowed. Therefore, $\left.  \Delta
E_{k}\right\vert _{\text{w}}=\Delta E_{k}=k^{2}(1/L_{\text{fn}}^{2}%
-1/L_{\text{in}}^{2})$. The partition function is given by $Z_{\text{in}%
}(\beta_{0})=%
{\textstyle\sum\nolimits_{k}}
e^{\beta_{0}E_{k,\text{in}}}$. It is trivially seen that Eq.
(\ref{GujratiRelation}) is satisfied, whereas the first equation in Eq.
(\ref{JarzynskiRelation}) due to Jarzynski fails in this case.

To conclude, we find that the change $\left.  \Delta\mathcal{H}\right\vert
_{\text{w}}$ in the Hamiltonian due to changes in its parameter is negative of
the work $\Delta W$ done \emph{by} the system, which contradicts the current
practice in diverse applications in nonequilibrium statistical thermodynamics
such as fluctuation theorems, quantum thermodynamics, stochastic
thermodynamics, etc. where $\left.  \Delta\mathcal{H}\right\vert _{\text{w}}$
is related to the work $\Delta\widetilde{W}$ done \emph{on} the system; see
the left hand and right hand sides in Eqs. (\ref{JarzynskiRelation}%
-\ref{Av-ExternalWork-Medium}). The correction ensures that Eq.
(\ref{GujratiRelation}) holds even for free expansion. We believe that the
correction requires complete reassessment of current applications.


\begin{thebibliography}{99}
\bibitem {note-0}Following modern notation \cite{Prigogine,deGroot}, exchanges
of any quantity $q(t)$ with the medium and changes within the system carry the
suffix e and i, respectively: $dq(t)\equiv d_{\text{e}}q(t)+d_{\text{i}}q(t)$.
For example, for the entropy $dS$, we have: $dS\equiv d_{\text{e}%
}S+d_{\text{i}}S$.
\end{thebibliography}
\end{document}